\documentclass[journal]{IEEEtran}

\usepackage[latin1]{inputenc}
\usepackage{graphicx}
\usepackage{amssymb,amsmath,amsfonts,amsthm}
\usepackage{algorithm,algorithmic}
\usepackage{cite}
\usepackage{float}
\usepackage{graphics}
\usepackage{subfigure}
\usepackage{xspace}
\usepackage[usenames,dvipsnames]{color}
\usepackage{amssymb, epsfig, hyperref}
\usepackage{mathtools}
\usepackage{bbm} 
\usepackage{subeqnarray}
\newtheorem{proposition}{Proposition}
\newtheorem{remark}{Remark}

\begin{document}

\title{Simultaneous Information and Energy Transfer in Large-Scale Networks with/without Relaying}

\author{Ioannis Krikidis,~\IEEEmembership{Senior Member,~IEEE}
\thanks{Manuscript received October 20, 2013; revised January 7, 2014. The editor coordinating the review of this paper and approving it for publication was E. Larsson.}
\thanks{I. Krikidis is with the Department of Electrical and Computer Engineering, Faculty of Engineering, University of Cyprus, Nicosia 1678 (E-mail: {\sf krikidis@ucy.ac.cy}).}
\thanks{This work was supported by the Research Promotion Foundation, Cyprus under the project KOYLTOYRA/BP-NE/0613/04 ``Full-Duplex Radio: Modeling, Analysis and Design (FD-RD)''.}}

\maketitle
%
%
%
%
\begin{abstract}
Energy harvesting (EH) from ambient radio-frequency (RF) electromagnetic waves is an efficient solution for fully autonomous and sustainable communication networks. Most of the related works presented in the literature are based on specific (and small-scale) network structures, which although give useful insights on the potential benefits of the RF-EH technology, cannot characterize the performance of general networks. In this paper, we adopt a large-scale approach of the RF-EH technology and we characterize the performance of a network with random number of transmitter-receiver pairs by using stochastic-geometry tools.  Specifically, we analyze the outage probability performance and the average harvested energy, when receivers employ power splitting (PS) technique for ``simultaneous'' information and energy transfer. A non-cooperative scheme,  where information/energy are conveyed only via direct links, is firstly considered and the outage performance of the system as well as the average harvested energy are derived in closed form in function of the power splitting. For this protocol, an interesting optimization problem which minimizes the transmitted power under outage probability and harvesting constraints, is formulated and solved in closed form. In addition, we study a cooperative protocol where sources' transmissions are supported by a random number of potential relays that are randomly distributed into the network. In this case, information/energy can be received at each destination via two independent and orthogonal paths (in case of relaying). We characterize both  performance metrics, when a selection combining scheme is applied at the receivers and a single relay is randomly selected for cooperative diversity. 
\end{abstract}


%
%
%
%

\section{Introduction}

Energy efficiency is of paramount importance for future  communication networks and is 
a main design target for all 5G radio access solutions. It refers to an efficient utilization of the available energy and consequently extends the network lifetime and/or reduces the operation cost.  Specifically,  conventional battery-powered communication systems suffer from short lifetime and require periodic replacement or recharging in order to maintain network connectivity. On the other hand, communication systems that are supported by a continuous power supply such as cellular networks require a power grid infrastructure and may result in large energy consumption that will further increase due to the increasing growth of data traffic. The investigation of energy-aware architectures as well as  transmission techniques/protocols that prolong the lifetime of the networks or provide significant energy savings has been a hot research area over several years, often under the umbrella of the green radio/communications \cite{HAN,ISM}.

Due to the limited supply of non-renewable energy resources, recently, there is a lot of interest to integrate the energy harvesting (EH) technology to power communication networks \cite{GUN2}. Energy harvesting is a new paradigm and allows nodes to harvest energy from natural resources (i.e., solar energy, wind, mechanical vibrations etc.) in order to maintain their operation. Related literature concerns the optimization of different network utility functions under various assumptions on the knowledge of the energy profiles. The works in \cite{YEN,GUN} assume that the EH profile is perfectly known at the transmitters and investigate optimal resource allocation techniques for different objective functions and network configurations. On the other hand, the works in \cite{EPH,KRI5} adopt a more networking point of view and maximize the performance in terms of stability region by assuming only statistical knowledge of the EH profile. Although energy harvesting from natural resources is a promising technology towards fully autonomous and self-sustainable communication networks, it is mainly unstable (i.e., weather-dependent) and thus less efficient for applications with critical quality-of-service (QoS) requirements.  

An interesting solution that overcomes the above limitation is to harvest energy from man-made electromagnetic radiation. Despite the pioneering work of Tesla, who experimentally demonstrated wireless energy transfer (WET) in late 19th century, modern wireless communication systems mainly focus on the information content of the radio-frequency (RF) radiation, neglecting the energy transported by the signals. Recently, there is a lot of interest to exploit RF radiation from energy harvesting perspective and use wireless energy transfer in order to power communication devices. The fundamental block for the implementation of this technology is the rectifying-antenna (rectenna) which is a diode-based circuit that converts the RF signals to DC voltage. Several rectenna architectures and designs have been proposed in the literature for different systems and frequency bands \cite{YAN,MON}. An interesting rectenna architecture is reported in \cite{VOL}, where the authors study a rectenna array in order to further boost the harvesting efficiency. Although information theoretic studies ideally assume that a receiver is able to decode information and harvest energy independently from the same signal \cite{GRO,FOU}, this approach is not feasible due to practical limitations. In the seminal work in \cite{RUI1}, the authors introduce two practical RF energy harvesting mechanisms for ``simultaneous'' information and energy transfer: a) time switching (TS) where dedicated time slots are used either for information transfer or energy harvesting, b) power splitting (PS) where one part of the received signal is used for information decoding, while the other part is used for RF energy harvesting.  

The employment of the above two practical approaches in different fundamental network structures, is a hot research topic and several recent works appear in the literature. In \cite{RUI1}, the authors study the problem of beamforming design for a point-to-point multiple-input multiple-output (MIMO) channel and characterize the  rate-energy region for both TS and PS techniques. This work is extended in \cite{XIA} for the case of an imperfect channel information at the transmitter by using robust optimization tools. The work in \cite{RUI2} investigates the optimal PS rule for a single-input single-output  (SISO) channel in order to achieve different trade-offs between ergodic capacity and average harvested energy. An interesting problem is discussed in \cite{RUI3}, where the downlink of an access point broadcasts energy to several users, which then use the harvested energy for  time-division multiple access (TDMA) uplink transmissions.  In \cite{TIM}, the authors study a fundamental multi-user multiple-input single-output (MISO) channel where the single-antenna receivers are characterized by both QoS and PS-EH constraints.  On the other hand,  cooperative/relay networks is a promising application area for  RF energy harvesting, since relay nodes have mainly limited energy reserves and may require external energy assistance. The works in \cite{KRI,NAS,DIN,HIM} deal with the integration of both TS and PS techniques in various  cooperative topologies with/without batteries for energy storage. The simultaneous information/energy transfer for a MIMO relay channel with a separated energy harvesting receiver is discussed in \cite{HIMSU}.

Although several studies deal with the analysis of communication networks with RF energy harvesting capabilities, most of existing work refers to specific (fixed) single/multiple user network configurations. 
Since harvesting efficiency is associated with the interference and thus the geometric distance between nodes, 
a fundamental question is to study RF energy harvesting for large-scale networks by taking into account random node locations. Stochastic-geometry is a useful theoretical tool in order to model the geometric characteristics of a large-scale network and derive its statistical properties \cite{STO, HANG,WEB1}.  Several works in the literature adopt stochastic-geometry in order to analyze the outage probability performance or the transmission capacity for different conventional (without harvesting capabilities) networks e.g., \cite{GAN,DIN3, MOH,HAN2}. Large-scale networks with energy harvesting capabilities are studied in \cite{RAHU,HARP,KWON,KAIB} for different network topologies and performance metrics. These works model the energy harvesting operation as a stochastic process and mainly refer to energy harvesting from natural resources e.g., solar, wind, etc. However, few studies  analyze the behavior of a RF energy harvesting network from a stochastic-geometry standpoint. In \cite{LEE}, the authors study the interaction between primary and cognitive radio networks,  where cognitive radio nodes can harvest energy from the primary transmissions, by modeling node locations as Poisson point processes (PPPs). A cooperative network with multiple transmitter-receiver pairs and a single energy harvesting relay is studied in \cite{DIN2} by taking the spatial randomness of user locations into consideration. The analysis of large-scale RF energy harvesting networks with practical TS/PS techniques, is an open question in the literature.

In this paper, we study the performance of a large-scale network with multiple transmitter-receiver pairs, where transmitters are connected to the power grid, while receivers employ the PS technique for RF energy harvesting.  By using stochastic-geometry, we model the randomness of node locations and we analyze the fundamental trade-off between outage probability performance and average harvested energy. Specifically, we study two main protocols: a) a non-cooperative protocol, and b) a cooperative protocol with orthogonal relay assistance.  In the non-cooperative protocol,  each transmitter simultaneously transfers information and energy at the associated receiver via the direct link. The outage probability of the system as well as the average harvested energy are derived in closed form in function of the power splitting ratio. In addition, an optimization problem which minimizes the transmitted power under some well-defined outage probability and average harvesting constraints, is discussed and closed form solutions are provided.

The cooperative protocol is introduced to show that relaying can significantly improve the performance of the system and achieve a better trade-off between outage probability performance and energy harvesting transfer. Relaying cooperation is integrated in several systems and standards in order to provide different levels of assistance (i.e., cooperative diversity, energy savings, secrecy etc); in this work, relays are used in order to facilitate the information/energy transfer. For the cooperative protocol, we introduce a set of potential dynamic-and-forward (DF) relays, which are randomly distributed in the network according to a PPP; these relays have similar characteristics with the transmitters and are also connected to the power grid. In this case, information and energy can be received at each destination via two independent paths (in case of cooperation). For the relay selection, we adopt a random selection policy based on a sectorized selection area with central angle at the direction of each receiver.  The outage performance of the system for a selection combining (SC) scheme as well as the average harvested energy are analyzed in closed form and validate the cooperative diversity benefits.  Numerical results for different parameter set-up reveal some important observations about the impact of the central angle and relay density on the trade-off between information and energy transfer.   It is the first time, to the best of the authors' knowledge, that stochastic-geometry is used in order to analyze a PS energy harvesting network with/without relaying.

The remainder of this paper is organized as follows. Section \ref{system} describes the system model and introduces the considered performance/harvesting metrics. Section \ref{direct_link} presents the non-cooperative protocol and analyzes its performance in terms of outage probability and average harvested energy. Section \ref{system2} introduces the cooperative protocol and analyzes both performance metrics considered. Simulation results are presented in Section \ref{num}, followed by our conclusions in Section \ref{conc}.

\underline{Notation:} $\mathbb{R}^d$ denotes the $d$-dimensional Euclidean space,  $\mathbf{1}(\cdot)$ denotes the indicator function, $|\cdot|$ is the Lebesgue measure, $b(x,r)$ denotes a two dimensional disk of radius $r$ centered at $x$, $\|x\|$ denotes the Euclidean norm of $x \in \mathbb{R}^d$, $\mathbb{P}(X)$ denotes the probability of the event $X$ and $\mathbb{E}(\cdot)$ represents the expectation operator. In addition, the number of points in $B$ is denoted by $N(B)$.

\section{System model}\label{system}

We consider a 2-D large-scale wireless network consisting of a random number of transmitter-receiver pairs. The transmitters form an independent homogeneous PPP $\Phi_t=\{x_k\}$ with $k\geq 1$ of intensity $\lambda$ on the plane $\mathbb{R}^2$, where $x_k$ denotes the coordinates of the node $k$. Each transmitter $x_k$ has a unique receiver $r(x_k)$ (not a part of $\Phi_t$) at an Euclidean distance $d_0$  in some random direction \cite{GAN}.  All nodes are equipped with single antennas and have equivalent characteristics and computation capabilities. The time is considered to be slotted and in each time slot all the sources are active without any coordination or scheduling process. In the considered topology, we add a transmitter $x_0$ at the origin $[0\; 0]$ and its associated receiver $r(x_0)$ at the location $[d_0\; 0]$  without loss of generality; in this paper, we analyze the performance of this typical communication link but our results hold for any node in the process $\Phi_t \cup x_0$ according to  Slivnyak's Theorem \cite{STO}. 

We assume a {\it partial fading} channel model, where desired direct links are subject to both small-scale fading and large-scale path loss, while interference links are dominated by the path-loss effects.  According to the literature \cite{HANG1,HANG2}, this channel model is denoted as ``1/0 fading'' and serves as a useful guideline for more practical configurations e.g., all links are subject to fading \cite{A,B}.  More specifically, the fading between $x_k$ and $r(x_k)$ is Rayleigh distributed so the power of the channel fading is an exponential random variable with unit variance.  The path-loss model assumes that the received power is proportional to $d(\mathcal{X},\mathcal{Y})^{-\alpha}$ where $d(\mathcal{X},\mathcal{Y})$ is the Euclidean distance between the transmitter $\mathcal{X}$ and the receiver $\mathcal{Y}$, $\alpha>2$ denotes the path-loss exponent and we define $\delta\triangleq\alpha/2$.  The Euclidean distance between two nodes is defined as

\begin{align}
d(\mathcal{X},\mathcal{Y})=\left\{ \begin{array}{l} \|\mathcal{X}-\mathcal{Y})  \|^{-\alpha}\;\;\text{If}\;\;\|\mathcal{X}-\mathcal{Y}  \|>r_0  \\ r_0^{-\alpha}\;\;\;\;\;\;\;\;\;\;\;\;\;\;\;\;\;\;\;\;\text{elsewhere},  \end{array}\right. \label{modelo_apostasis}
\end{align}

\noindent where the parameter $r_0>1$ refers to the minimum possible path-loss degradation and ensures the accuracy of our path-loss model for short distances \cite{HAN2}. The instantaneous fading channels are known only at the receivers in order to perform coherent detection.  In addition, all wireless links exhibit additive white Gaussian noise (AWGN) with variance $\sigma^2$.

The transmitters are continuously connected to a power supply (e.g., battery or power grid) and transmit with the same power $P_t$. On the other hand,  each receiver has RF energy harvesting capabilities and can harvest energy from the received electromagnetic radiation. The RF energy harvesting process is based on the PS technique and therefore each receiver splits its received signal in two parts a) one part is converted to a baseband signal for further signal processing and data detection (information decoding) and b) the other part is driven to the rectenna for conversion to DC voltage and energy storage. Let $\nu_d \in (0,1)$ denote the power splitting parameter for each receiver; this means that $100 \nu_d \%$ of the received power is used for data detection while the remaining amount is the input to the RF-EH circuitry. We assume an ideal power splitter at each receiver without power loss or noise degradation, and that the receivers can perfectly synchronize their operations with the transmitters based on a given power splitting ratio $\nu_d$ \cite{ZHOU1}. During the baseband conversion phase, additional circuit noise, $v$, is present due to phase-offsets and circuits' non-linearities and  which is modeled as AWGN with zero mean and variance $\sigma_C^2$ \cite{RUI1}. Based on the PS technique considered, the signal-to-interference-plus-noise ratio (SINR) at the typical receiver can be written as
\begin{align}
{\sf SINR}_0=\frac{\nu_d P_t h_0 d_0^{-\alpha}}{\nu_d \big(\sigma^2+P_t I_0 \big)+\sigma_C^2}, \label{SNR}
\end{align}

\noindent where $I_0 \triangleq\sum\limits_{x_k \in \Phi_t}d(x_k)^{-\alpha}$ denotes the total (normalized) interference at the typical receiver with $d(x_k)\triangleq d \big(x_k,r(x_0) \big)$ and $h_k$ denotes the channel power gain for the link $x_k\rightarrow r(x_k)$. A successful decoding requires that the received SINR is at least equal to a detection threshold $\Omega$. On the other hand, RF energy harvesting is a long term operation\footnote{In most real-world applications, the received power is very low (scale of dBm); therefore instantaneous harvesting has not practical interest.} and is expressed in terms of average harvested energy \cite{RUI1,RUI2}. Since $100(1-\nu_d)\%$ of the received energy is used for rectification, the average energy harvesting  at the typical receiver is expressed as
\begin{align}
E_0=\zeta\cdot \mathbb{E}\bigg((1-\nu_d)P_t\big[h_0d_0^{-\alpha}+I_0\big] \bigg), \label{mean}
\end{align}

\noindent where $\zeta \in (0,1]$ denotes the conversion efficiency from RF signal to DC voltage; for the convenience of analysis, it is assumed that $\zeta=1$. It is worth noting that the RF energy harvesting from the AWGN noise is considered to be negligible. 

\section{Non-cooperative protocol for simultaneous information/energy transfer}\label{direct_link}

The first investigated scheme does not enable any cooperation between the nodes and thus communication is performed in a single time slot; all the sources simultaneously transmit towards their associated receivers. Fig. \ref{model1} schematically presents the network topology for the non-cooperative case.  The information decoding process is mainly characterized by the outage probability which denotes the probability that the instantaneous SINR is lower than the predefined threshold $\Omega$.  By characterizing the outage probability for the typical transmitter-receiver link $x_0\rightarrow r(x_0)$, we also characterize the outage probability for whole network ($\forall\; x_k\in \{\Phi_t\cup x_0\}$).

\begin{figure}[t]
\centering
\includegraphics[width=\linewidth]{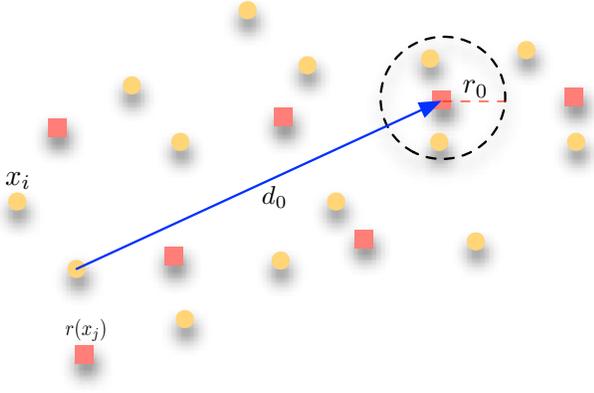}
\vspace{-0.5cm}
\caption{Network topology for the non-cooperative protocol; $d_0$ is the Euclidean distance between a transmitter and its associated receiver and $r_0$ is the minimum path-loss distance.}
\label{model1}
\end{figure}

\begin{proposition}\label{prop1}
The outage probability for the non-cooperative protocol is given by
\begin{align}
\Pi_{\text{NC}}(\nu_d,P_t)&=\mathbb{P}\{{\sf SINR}_0<\Omega \} \nonumber\\
&=1-\exp \left( -\frac{\Omega d_0^{\alpha} \sigma^2}{P_t}-\frac{\Omega d_0^{\alpha} \sigma_C^2}{\nu_d P_t}  \right)\Xi(\lambda,d_0,r_0), \label{pout}
\end{align}
\noindent where 
\begin{align}
\Xi(x,y,z)\triangleq&\exp \Bigg(-\pi x \bigg[ \left(\exp \left(-\Omega y^{\alpha} z^{-\alpha} \right)-1\right) z^2 \nonumber \\
&+\Omega^\delta y^2 \gamma \left(1-\delta,\Omega y^{\alpha} z^{-\alpha} \right) \bigg] \Bigg) \nonumber \\
&\times\exp \left(-2 \pi \Omega x  y^{\alpha} z^{2-\alpha} \right).
\end{align}
\end{proposition}

\begin{proof} 
See Appendix \ref{app1}.
\end{proof}

For high transmitted powers  i.e., $P_t\rightarrow \infty$ and $\nu_d>0$ the system becomes interference limited and the outage probability converges to a constant error floor given by
\begin{align}\label{high_exp}
\Pi_{\text{NC}}^{\infty}\rightarrow 1-\Xi(\lambda,d_0,r_0).
\end{align}
As for the average harvested energy,  by expanding \eqref{mean} we have the following proposition: 

\begin{proposition}\label{prop2}
The average harvested energy for the non-cooperative protocol is given by
\begin{align}
E_{NC}=(1-\nu_d)P_t \big[ d_0^{-\alpha}+\Psi(\lambda) \big], \label{prop2e}
\end{align}
\noindent where

\begin{align}
\Psi(x)\triangleq \pi x r_0^{2-\alpha}\frac{\alpha}{\alpha-2}.
\end{align}
\end{proposition}

\begin{proof}
From \eqref{mean}, we have: 
\begin{align}
E_{NC}&=\mathbb{E}\bigg((1-\nu_d)P_t\big[ h_0 d_0^{-\alpha}+I_0\big] \bigg) \nonumber \\
&=(1-\nu_d)P_t \bigg( d_0^{-\alpha} \mathbb{E}(h_0)+\mathbb{E}(I_0)\bigg) \nonumber \\
&=(1-\nu_d)P_t \big[ d_0^{-\alpha}+\Psi(\lambda) \big],
\end{align}
\noindent where $\mathbb{E}(h_k)=1$ for all $k$, and the proof of $\mathbb{E}(I_0)=\Psi(\lambda)$ can be found in Appendix \ref{app2}.
\end{proof}

\subsubsection{Optimization problem- minimum transmitted power} \label{optimization_sec}

An interesting optimization problem is formulated when energy becomes a critical issue for the network and each receiver is characterized by both QoS and RF energy harvesting constraints. Due to the symmetry of the nodes, the minimization of the transmitted power for the typical transmitter, it also minimizes the total energy consumption for whole network. The optimization problem considered can be written as
\begin{align}
&\min_{P_t,\nu_d} P_t \nonumber \\
&\text{subject to}\; \Pi_{\text{NC}} \leq C_I, \nonumber \\
&\;\;\;\;\;\;\;\;\;\;\;\;\;\;\;\;E _{\text{NC}}\geq C_H, \nonumber \\
&\;\;\;\;\;\;\;\;\;\;\;\;\;\;\;\;0\leq \nu_d \leq 1, \nonumber \\
&\;\;\;\;\;\;\;\;\;\;\;\;\;\;\;\;P_t\geq 0, \label{optimization}
\end{align}
\noindent where the QoS constraint ensures an outage probability lower than a threshold $C_I$, while the RF energy harvesting constraint requires an average harvested energy at least equal to  $C_H$ (i.e., it represents the minimum required energy to maintains operability at each device). For the case where the power splitting ratio is constant i.e., $\nu_d=\nu_0$, the solution to the optimization problem in \eqref{optimization} is simplified as follows:

\begin{align}
P_t^*=\left\{ \begin{array}{l} \max\left[\frac{G_1}{(1-\nu_0)}, \left(G_2+\frac{G_3}{\nu_0} \right)/G_0 \right] \;\text{If}\; \Pi_{\text{NC}}^{\infty}\leq C_I \\
\text{Infeasible},\;\;\text{elsewhere} \end{array} \right.
\end{align}

\noindent where $G_0\triangleq \ln\left(\frac{\Xi(\lambda,d_0,r_0)}{1-C_I}\right)$, $G_1\triangleq C_H/[d_0^{-\alpha}+\Psi(\lambda)]$, $G_2\triangleq \Omega d_0^2 \sigma^2$ and $G_3\triangleq\Omega d_0^2 \sigma_C^2$. The asymptotic expression in \eqref{high_exp}  is involved in the optimization problem and determines its feasibility. More specifically, if the outage probability floor in \eqref{high_exp} is higher than the outage probability constraint $C_I$, there is not any transmitted power that can satisfy $C_I$ and therefore the optimization problem becomes infeasible. The constant power splitting case corresponds to a low implementation complexity and is appropriate for (legacy) systems where the rectenna's design is predefined and the power splitting parameter is not adaptable. For the general case, where both $P_t$ and $\nu_d$ are adjustable, it can be easily seen that the two main constraints are binding at the solution\footnote{Problem in \eqref{optimization} requires at least one of the constraints to be binding, otherwise the value of $P_t$ can further be reduced. By examining the cases where one constraint is binding and the other holds with inequality, we show that at the optimal solution the inequality constraint holds with equality.}. In this case, the optimization problem is transformed to the solution of a standard quadratic equation and for $\Pi_{\text{NC}^{\infty}}\leq C_I$ the solution is given by 
\begin{align}
&P_t^*=\frac{G_1}{(1-\nu_d^*)}, \nonumber \\
&\nu_d^*=\frac{-(G_0G_1+G_2-G_3)+\sqrt{(G_0G_1+G_2-G_3)^2+4G_2G_3}}{2G_2}. 
\end{align}
\noindent We note that the optimization problem is infeasible  for $P_{\text{NC}}^{\infty}> C_I$. The general optimization problem requires adaptive and dynamic RF power splitting and therefore refers to a higher implementation complexity.  

The implementation problem can be solved either by a central controller or in a distributed fashion. In the first case, a central unit that controls the network, solves the problem and broadcasts the common solution (transmitted power, power splitting ratio) to all nodes; the transmitters and the receivers adjust their transmitted power and the power splitting ratio, respectively. In the distributed implementation, each node can locally solve the optimization problem without requiring external signaling (but with the cost of a higher computational complexity). The optimization problem involves only deterministic and average system parameters such as geometric distances, network density, path-loss exponent and channel statistics; these parameters are estimated at the beginning of the communication and remain constant for a long operation time.

\begin{figure}[t]
\centering
\includegraphics[width=0.85\linewidth]{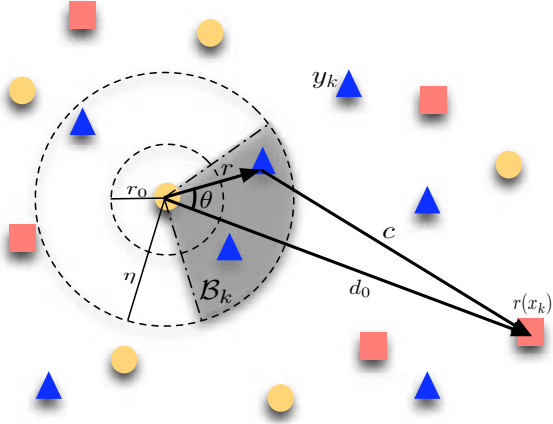}
\vspace{-0.5cm}
\caption{Network topology for the cooperative protocol; $\eta$ is the radius of the selection sector.}
\label{model2}
\end{figure}

\section{Cooperative protocol for simultaneous information/energy transfer}\label{system2}

The cooperative scheme exploits the relaying/cooperative concept in order to combat fading and  path-loss degradation effects. The network topology considered is modified by adding a group of single-antenna DF relays, which have not their own traffic and are dedicated to assist the transmitters. Fig. \ref{model2} schematically presents the network topology for the cooperative protocol. The location of all relay nodes are modeled as a homogeneous PPP denoted by $\Phi_r=\{y_k\}$ with density $\lambda_r$; this assumption refers to mobile relays where their position as well as their ``availability'' changes with the time. The relay nodes are also continuously connected to a power supply (e.g., battery) and have equivalent computation/energy capabilities.  We adopt an orthogonal relaying protocol where cooperation is performed in two orthogonal time slots \cite{GAN,MOH}. It is worth noting that although several cooperative schemes have been proposed in the literature e.g., \cite{AZA}, the orthogonal relaying protocol has a low complexity and is sufficient for the purposes of this work. The cooperative protocol operates as follows:

\begin{enumerate}
\item The first phase of the protocol is similar to the non-cooperative scheme and thus all transmitters simultaneously broadcast their signals towards the associated receivers. Each transmitter  $x_k$ defines a 2-D relaying area $\mathcal{B}_k$ around its location and each relay node located inside this area is dedicated to assist this transmitter; this means that all relays $y_i \in \mathcal{B}_k$ consider the signal generated by $x_k$ as a useful information and all the other signals as interference.  In accordance to the general system model, we assume that direct links suffer from both small-scale fading and path-loss, while interference links are dominated by the path-loss attenuation  (1/0 partial fading \cite{HANG1}). By focusing  our study on the typical transmitter $x_0$, we define as $g_k$ the power fading gain for the link $x_0\rightarrow y_k$  with $y_k \in \mathcal{B}_0$. In this case, the direct link $x_0\rightarrow r(x_0)$ is characterized  by  \eqref{SNR}, \eqref{mean}, while the SINR at the relay $y_k$ is written as

\begin{align}
{\sf SINR}_k=\frac{P_{t}g_kd(x_0,y_k)^{-\alpha}}{\sigma^2+P_t I_k},
\end{align}     
       
\noindent where $I_k=\sum_{x\in \Phi_t}d(x,y_k)^{-\alpha}$ denotes the total (normalized) interference received at $y_k$. If the relay node $y_k$ can decode the transmitted signal, which means that ${\sf SINR}_k\geq \Omega$, it becomes a member of the transmitter's potential relay. It is worth noting that the relay nodes use all the received signal for information decoding, since they have not energy harvesting requirements. 

\item In the second phase of the protocol, one relay node (if any) that successfully decoded the transmitted signal accesses the channel and retransmits the source's signal. We assume a {\it random selection} process which selects a single relay out of all potential relays with equal probability. The random relay selection does not require any instantaneous channel feedback or any instantaneous knowledge of the geometry and is appropriate for low complexity implementations with strict energy constraints  \cite{MOH}. More sophisticated relay selection policies, which take into account the instantaneous channel conditions \cite{GAN,NGHI}, can also be considered in order to further improve the cooperative benefits. We define as $y^*$ the selected relay for the typical transmitter $x_0$, $f$ is the channel power gain for the link $y^*\rightarrow r(x_0)$ and $\Phi_{r}^*$ is the homogeneous PPP that contains all the selected relays for whole network.  If the potential relay set is empty for a specific transmitter (no relay was able to decode the source's  transmitted signal), its message is not transmitted during the second phase of the protocol and therefore does not enjoy cooperative diversity benefits. The relaying link for the (typical) receiver is characterized by the following equations

\begin{align}
&{\sf SINR}_0'=\frac{b \nu_r P_r f d\big(y^*,r(x_0) \big)^{-\alpha}}{\nu_r \big(\sigma^2+P_r I_0' \big)+\sigma_C^2},\;\;\;\text{when}\; b=1, \\
&E_0'=\zeta\cdot \mathbb{E}\bigg((1-\nu_r)P_r \big[b fd(y^*,r(x_0))^{-\alpha}+I_0' \big] \bigg), \label{mean2}
\end{align}

\noindent where $I_k'=\sum_{y \in \Phi_r^*}d(y,r(x_k))^{-\alpha}$ denotes the total (normalized) interference at the receiver $r(x_k)$, $\nu_r \in (0,\;1)$ is the power splitting ratio used in the second phase of the protocol, $P_r$ is the transmitted power for each active relay and the binary variable $b \in \{0,1\}$ is equal to one in case of a relaying transmission, while it takes the value zero when the relay set is empty. A perfect synchronization between the selected relay and the associated receiver is assumed for a given power splitting ratio $\nu_r$. As for the decoding process at the receivers, we assume that the two copies of the transmitted signal are combined with a simple SC technique; this means that information decoding is based on the best path between direct/relaying links \cite{NGHI}. It is well known that SC only requires relative SINR measurements and thus it is simpler than maximum ratio combiner, which requires exact knowledge of the channel state information for each diversity branch. In addition, SC significantly reduces power consumption because continuous estimates of the channel state information are not necessary; this is beneficial for the considered RF energy harvesting system, where energy saving is a critical requirement e.g., \cite{HUBE,CHAU,SELV}.  Regarding the potential use of the non-selected branch for RF energy harvesting, here, we assume a simple/conventional implementation and the received energy allocated to the SC cannot be used for RF energy harvesting purposes. On the other hand, the energy harvesting process exploits both transmission phases and the total average energy harvested becomes equal to
\begin{align}
E_{\text{CO}}=E_0+E_0'.\label{total}
\end{align}
\end{enumerate}

The considered random relay selection process does not require any instantaneous channel feedback and is appropriate for scenarios with critical energy/computation constraints. However, the definition of the  selection area $\mathcal{B}_k$ has a significant impact on the system performance. In this work, we assume that  $\mathcal{B}_k$ is a circular sector\footnote{The considered cooperative protocol assumes that the selection sectors have not overlaps and therefore a relay can be inside into a single selection sector.  Although this assumption simplifies our analysis, the work in \cite[Sec. II.B]{LEE} shows that for practically small $\lambda$ and $\eta$,  circular discs around the different transmitters do not overlap at most of the time. In our case, we have circular sectors and therefore the probability of overlapping becomes much lower.} with center $x_k$, radius $\eta>r_0$ and central angle with orientation at the direction of the receiver $r(x_k)$. 

\begin{remark}
By appropriately adjusting the central angle of the sector, we can ensure that the relaying paths are shorter than the direct distance $d_0$; this parameterization avoids scenarios where the selected relay experiences more serious path-loss effects than the direct link. It is proven in Appendix \ref{angle_sec} that a selection area $\mathcal{B}_k=\big\{r \in [0\; \eta],\; \theta \in [-\theta_0\;\theta_0]\big\}$ with  $\theta_0\leq \cos^{-1}(\eta/(2d_0))$ satisfies this requirement. 
\end{remark}

\subsection{Outage probability and average harvested energy}

For the cooperative protocol, an outage event occurs when (a)  the direct link is in outage and no relay is able to decode the source's message or (b) a relay node is able to decode the transmitted signal but both direct and relaying link are in outage. Based on these two cases, we have the following proposition: 

\begin{proposition}
The outage probability for the cooperative protocol is given by 
\begin{align}
\Pi_{\text{CO}}(\nu_d,\nu_r,P_t,P_r)&=\underbrace{\Pi_{\text{NC}}(\nu_d,P_t)\cdot  \Pi_c(P_t)}_{\text{case (a)}} \nonumber \\
&+\underbrace{ \big(1-\Pi_c(P_t) \big)\cdot \Pi_{\text{NC}}(\nu_d,P_t)\cdot \Pi_r(\nu_r,P_r)}_{case (b)} \label{out_expr}
\end{align}
where
\begin{align}
&\Pi_c(P_t) \nonumber \\
&=\exp\left(-\lambda_r\left[\int_{-\theta_0}^{+\theta_0}\int_{r_0}^{\eta} \exp\left(-\frac{\sigma^2 \Omega r^{\alpha}}{P_t} \right)\Xi(\lambda,r,r_0)r dr d\theta \right. \right. \nonumber \\
&\;\;\;\;\;\;\;\;\;\;\;\;\;\;\;\;\;\;\;\;\;\;\left. \left. +\theta_0 r_0^2 \exp\left(-\frac{\sigma^2 \Omega r_0^{\alpha}}{P_t} \right)\Xi(\lambda,r_0,r_0) \right]  \right), \nonumber \\
&\rightarrow \underbrace{\exp\left(-\lambda_r\left[\int_{-\theta_0}^{+\theta_0}\int_{r_0}^{\eta} \Xi(\lambda,r,r_0)r dr d\theta+\theta_0 r_0^2 \Xi(\lambda,r_0,r_0) \right]  \right)}_{\Pi_c^{\infty},\;\text{for}\; P_t\rightarrow \infty  },\label{prob_fst} \\
&\Pi_r(\nu_r,P_r)=1-\frac{1}{\theta_0(\eta^2-r_0^2)}\int_{-\theta_0}^{\theta_0}\int_{r_0}^{\eta}\exp\left(-\frac{\sigma^2 \Omega c^{\alpha}}{P_r} \right) \nonumber \\
&\;\;\;\;\;\;\;\times \exp\left(-\frac{\sigma_C^2 \Omega c^{\alpha}}{\nu_r P_r} \right) \Xi(\lambda [1-\Pi_c], c, r_0)r dr d\theta,  \nonumber \\
&\;\;\;\rightarrow \underbrace{1-\frac{1}{\theta_0(\eta^2-r_0^2)}\int_{-\theta_0}^{\theta_0}\int_{r_0}^{\eta}\Xi(\lambda [1-\Pi_c], c, r_0)r dr d\theta}_{\Pi_r^{\infty},\;\text{for}\; P_r\rightarrow \infty},
\end{align}
\end{proposition}
\noindent with $c=\sqrt{r^2+d_0^2-2rd_0\cos(\theta)}$.
\begin{proof}
See Appendix \ref{app_c1} for the outage probability  of the first hop ($\Pi_c(\cdot)$)  and Appendix \ref{app_c2} for the outage probability  of the second (relaying) hop ($\Pi_r(\cdot)$).
\end{proof}

We note that for the case where $P_t,P_r\rightarrow \infty$ or $\sigma^2,\sigma_C^2\rightarrow 0$ with $\nu_d,\nu_r>0$, the system becomes interference limited and the outage performance converges to a constant outage probability floor given by
\begin{align}
\Pi_{\text{CO}}^{\infty}\rightarrow \Pi_{\text{NC}}^{\infty}\Pi_c^{\infty}+(1-\Pi_c^{\infty})\Pi_{\text{NC}}^{\infty}\Pi_r^{\infty}.
\end{align}
On the other hand, each receiver harvests energy from both phases of the cooperative protocol. In contrast to the information decoding, which highly depends on the relaying transmission and thus becomes inactive  in case of an empty relay set, the RF harvesting process is active in all cases. More specifically, in case where the relay set is empty (no relay reforwards the source's message), the corresponding receiver does not employ a PS technique and uses all the received energy (e.g., interference) for RF energy harvesting.  Based on this fundamental remark, we have the following proposition

\begin{proposition}
The average harvested energy for the cooperative protocol is given by
\begin{align}
E_{\text{CO}}&=(1-\nu_d)P_t \big[d_0^{-\alpha}+\Psi(\lambda) \big]+\Pi_c \cdot P_r \Psi \big(\lambda[1-\Pi_c] \big) \nonumber \\
&\;\;\;\;\;+(1-\Pi_c) \cdot(1-\nu_r)P_r \bigg[\mathcal{Z}+\Psi(\lambda [1-\Pi_c]) \bigg],
\end{align}
where $\mathcal{Z}\approx\bigg(\left(\frac{r_0+\eta}{2}\right)^2+d_0^2-2d_0 \frac{r_0+\eta}{2}\cdot \frac{\sin(\theta_0)}{\theta_0}  \bigg)^{-\delta}$ denotes the average attenuation for the relaying link. 
\end{proposition}

\begin{proof}
From \eqref{total}, we have 

\begin{align}
E_{\text{CO}}&=E_{0}+E_0' \nonumber \\
&=\underbrace{E_{\text{NC}}}_{\text{direct link}}+\underbrace{\Pi_c\cdot  \mathbb{E}(P_r I_0')}_{\text{relaying link is inactive}} \nonumber \\
&\;\;\;\;\;+\underbrace{\big(1-\Pi_c\big) \cdot \mathbb{E} \big(P_r d(y^*)^{-\alpha}+P_r I_0' \big)}_{\text{relaying link is active}} \\
&=(1-\nu_d)P_t \big[d_0^{-\alpha}+\Psi(\lambda) \big]+\Pi_c\cdot P_r \Psi \big(\lambda [1-\Pi_c] \big) \nonumber \\
&\;\;\;\;\;+(1-\Pi_c)\cdot (1-\nu_r)P_r \bigg[\mathcal{Z}+\Psi(\lambda [1-\Pi_c]) \bigg],
\end{align}

\noindent where the proof of $\mathbb{E}(d(y^*)^{-\alpha})\approx \mathcal{Z}$ is reported in Appendix \ref{app_aver}; it is worth noting that Appendix \ref{app_aver} provides both the exact value of the average attenuation as well as the above simplified approximation. 
\end{proof}

It is worth noting that a similar optimization problem with \eqref{optimization} can be formulated for the relaying case; the objective function could be the minimization of the total transmitted power i.e., $\min P_t+P_r$. However, as it can be seen from Proposition 3, the expressions of the outage probability are complicated in the relaying case and do not allow elegant closed form solutions for the optimization problem.

\section{Numerical results}\label{num}
Computer simulations are carried out in order to evaluate the performance of the proposed schemes. The simulation environment follows the description in Sections \ref{system}, \ref{system2} with parameters\footnote{Unless otherwise defined.} $d_0=20$ distance units (we will use meters (m) for the sake of presentation), $r_0=4$ m, $\nu_d=\nu_r=0.3$,  $\sigma^2=\sigma_C^2=1$, $\Omega=-30$ dB, $\alpha=4$ and $\zeta=1$; the average harvested energy is measured in Watts. For the sake of simplicity, we assume that the relay nodes have similar computational/complexity characteristics with the transmitters (e.g., they could be inactive transmitters of the network) and therefore transmit with a power $P_r=P_t$.  For the cooperative protocol we assume $\eta=8$ m and thus Remark 1 corresponds to $\theta_0\leq \cos^{-1}(1/5)=0.4359 \pi$. The presented results concern the typical link $x_0\rightarrow r(x_0)$ but refer to any link of the network $\Phi_t\cup x_0$ (according to the Slivnyak's Theorem \cite[Sec. 8.5]{HAN2}).

\begin{figure}[t]
\centering
\subfigure[Outage probability.]{
  \includegraphics[width=\linewidth]{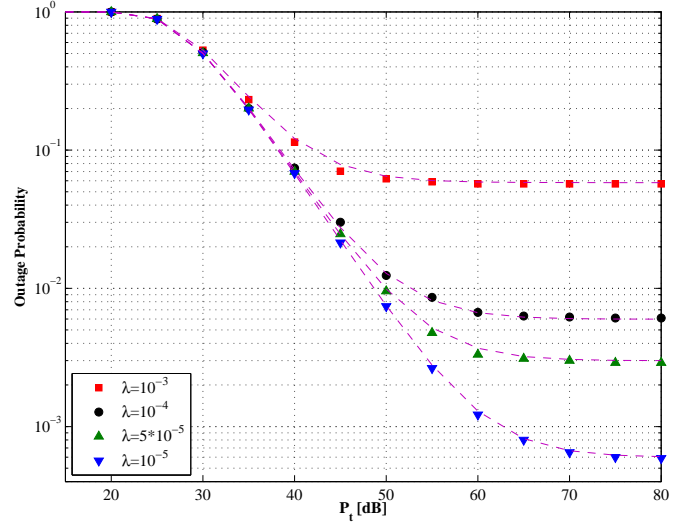}
  \label{fig1a}
 }
 \subfigure[Mean harvested energy.]{
 \includegraphics[width=\linewidth]{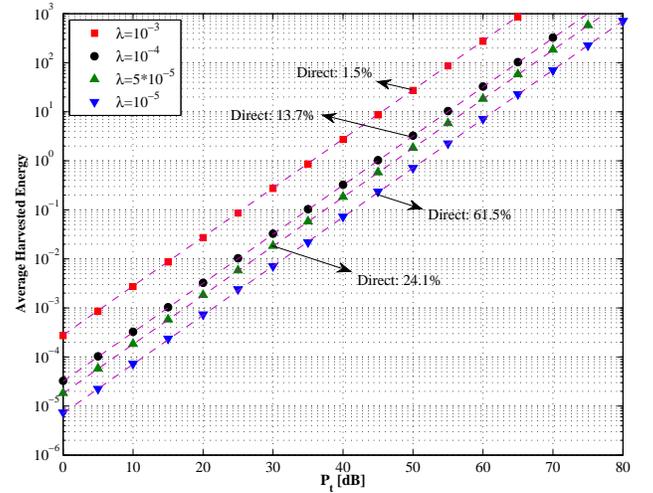}
   \label{fig1b}
 }
\vspace{-0.3cm}
\caption{Performance of the non-cooperative protocol versus $P_t$ for different network densities $\lambda$;  
$\sigma^2=\sigma_C^2=1$, $\Omega=-30$ dB, $r_0=4$ m, $d_0=20$ m, $\nu_d=0.3$ and $\alpha=4$.
Analytical results are shown with dashed lines.}
\end{figure}

\subsection{Non-cooperative protocol}

Fig.'s \ref{fig1a}, \ref{fig1b}  deal with the performance of the non-cooperative scheme for different network densities e.g.,  $\lambda=\{10^{-5}, 5\times 10^{-5}, 10^{-4}, 10^{-3} \}$. Specifically, Fig. \ref{fig1a} plots the outage probability of the system versus the transmitted power $P_t$. The first main observation is that the outage performance converges to a constant floor for high transmitted powers $P_t\rightarrow \infty$. This behavior is due to the fact that all nodes transmit with the same power without any coordination (scheduling) and therefore the system becomes interference limited as $P_t$ increases. As for the impact of the network density on the outage performance, it can be seen that as the density increases, the outage probability of the system increases; for $\lambda=10^{-5}$ the outage probability converges to $6\times 10^{-4}$, while for $\lambda=10^{-3}$ it converges to $6\times 10^{-2}$. This observation shows that the network density and the related multi-user interference significantly affects the decoding ability of the receivers. In the same Figure, we plot the analytical results given by \eqref{pout}, which  match with the simulation results and validate our analysis. On the other hand,  Fig. \ref{fig1b} plots the RF average harvested energy versus the transmitted power $P_t$. It can be seen that the RF average harvested energy is a linear function of the transmitted power (as it can be observed by \eqref{prop2e}) and increases as the transmitted power increases. By comparing the different curves, we can see that as the network density increases, the average harvested energy increases; e.g.,  if $P_t=45$ dB,  we have $E_{\text{NC}}=1$ Watt for $\lambda=10^{-4}$ and $E_{\text{NC}}=9$ Watt for  $\lambda=10^{-3}$. This remark shows that a dense network facilitates the RF energy harvesting process and thus interference is beneficial from an energy harvesting standpoint. In Fig. \ref{fig1b}, we also show the percentage of the harvested energy which is from the direct link. As it can be seen, for small network densities, the direct component significantly contributes to the total average harvested energy while becomes less important as the network density increases. For high network densities i.e., $\lambda=10^{-3}$, interference dominates the RF energy harvesting process and the percentage of the direct link is almost negligible i.e., $1.5$ \%. The theoretical curves perfectly match with the simulation results and validate our analysis in Proposition \ref{prop2}. It is worth noting that these two figures demonstrate the fundamental trade-off between information decoding and RF energy harvesting; interference significantly degrades the achieved outage performance, while it becomes helpful for the RF energy harvesting process.

\begin{figure}[t]
\centering
\subfigure[Outage probability.]{
  \includegraphics[width=\linewidth]{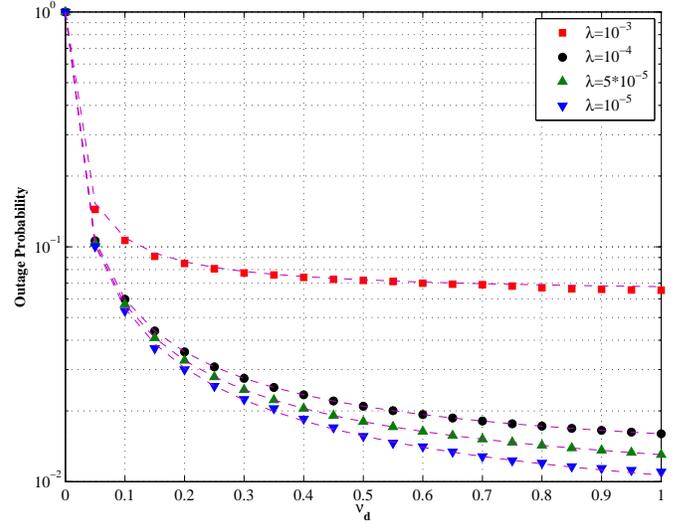}
  \label{fig2a}
 }
 \subfigure[Mean harvested energy.]{
 \includegraphics[width=\linewidth]{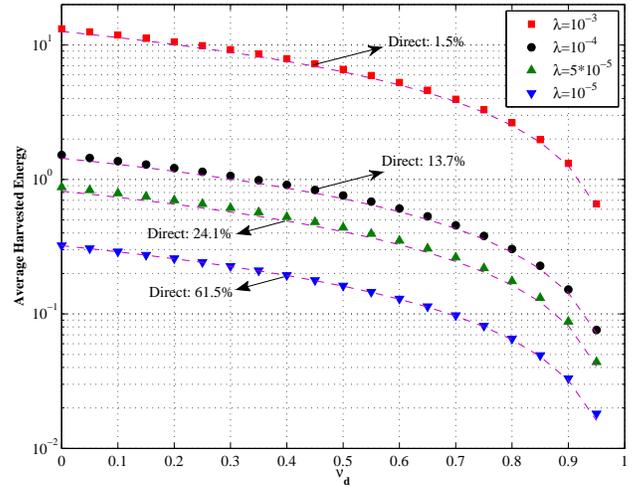}
   \label{fig2b}
 }
\label{figure2}
\vspace{-0.3cm}
\caption{Performance of the non-cooperative protocol versus $\nu_d$ for different network densities $\lambda$;  
$\sigma^2=\sigma_C^2=1$, $\Omega=-30$ dB, $r_0=4$ m, $d_0=20$ m, $P_t=45$ dB and $\alpha=4$.
Analytical results are shown with dashed lines.}
\end{figure}

\begin{table*}
    \centering
	    \caption{Optimal transmitted power (in Watt) for the non-cooperative protocol; $\lambda=10^{-5}$, $\sigma^2=\sigma_C^2=1$, $d_0=20$ m, $r_0=4$ m, $\Omega=-30$ dB and $\alpha=4$.}
	    \label{Tab}
        \begin{tabular}{|c| c| c| c |c|}
            \hline
        $\begin{array}{l} C_\text{I}=10^{-3}\\ C_{\text{H}}=10^3 \end{array}$ & Transmitted power $P_t$ & Power splitting $\nu_d$ & Outage Probability & Average harvested energy \\
        \hline
         Optimization I & $1.9652\times 10^{8}$ & $0.5$  & $6.0153\times 10^{-4}$ & $10^{3}$ \\
         \hline
          Optimization II  & $ 9.8661\times 10^{7}$ & $0.0041$  & $10^{-3}$ & $10^{3}$ \\
          \hline
          \hline
        \end{tabular}
        \begin{tabular}{|c| c| c| c |c|}
            \hline
        $\begin{array}{l} C_\text{I}=0.01\\ C_{\text{H}}=0.1 \end{array}$ & Transmitted power $P_t$ & Power splitting $\nu_d$ & Outage Probability  & Average harvested energy \\
        \hline
         Optimization I & $  5.0788\times 10^{4}$ & $0.5$  & $0.01$ & $0.2584$ \\
         \hline
          Optimization II  & $3.9470\times 10^{4}$ & $0.7511$  & $0.01$ & $0.1$ \\
          \hline
        \end{tabular}
\end{table*}

Fig.'s \ref{fig2a}, \ref{fig2b} show the impact of the power splitting ratio $\nu_d$ on the outage performance and the RF average harvested energy, respectively. We assume $P_t=45$ dB and the other parameters are similar to the previous simulation example. A higher $P_t$ affects (decreases) the outage probability in accordance with Fig.  \ref{fig1a} but does not change the observed behavior of the outage probability versus $\nu_d$ curves; therefore further simulation results with other $P_t$ parameters do not add value to the main remarks of the paper. It can be seen that the power splitting ratio significantly affects the performance of the system and defines the balance between the two conflicting objectives i.e., outage probability Vs energy harvesting. Specifically,  as $\nu_d$ increases the outage performance is improved while the harvesting process becomes less efficient, since most of the received energy is used for information decoding;  when $\nu_d$ decreases we have the inverse behavior, since most of the received energy is used for RF-to-DC rectification. Regarding the network's density $\lambda$, our observations confirm the previous main remarks.

In Table \ref{Tab}, we deal with the optimization problem discussed in Section \ref{optimization_sec}. Specifically, we present the optimal solution for the case of a constant splitting power ratio $\nu_0=0.5$ ({\it Optimization I}) as well as for the case where both parameters $P_t,\nu_d$ can be adjusted ({\it Optimization II}). The first observation is that {\it Optimization II} gives a solution $P_t^*$ which is much lower than this one of {\it Optimization I}, since the power spitting parameter is optimized accordingly. The solution of {\it Optimization II} satisfies both constraints with equality, as it has been discussed in  Section \ref{optimization_sec}. On the other hand, for the first set-up $(C_{\text{I}}=10^{-3},C_{\text{H}}=10^3)$, the solution of the {\it Optimization I} satisfies the harvesting constraint with equality, since $C_{\text{H}}$ is the dominant constraint. For the setup $(C_{\text{I}}=10^{-2},C_{\text{H}}=10^{-1})$,  the optimal solution satisfies the outage constraint with equality, since $C_{\text{I}}$ becomes the dominant constraint for this case. 

\begin{figure}[h!]
\centering
\subfigure[Outage probability.]{
  \includegraphics[width=\linewidth]{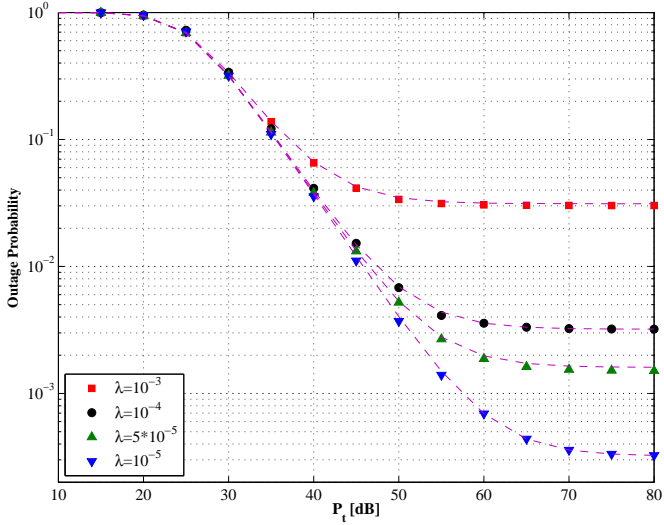}
  \label{fig3a}
 }
 \subfigure[Mean harvested energy.]{
 \includegraphics[width=\linewidth]{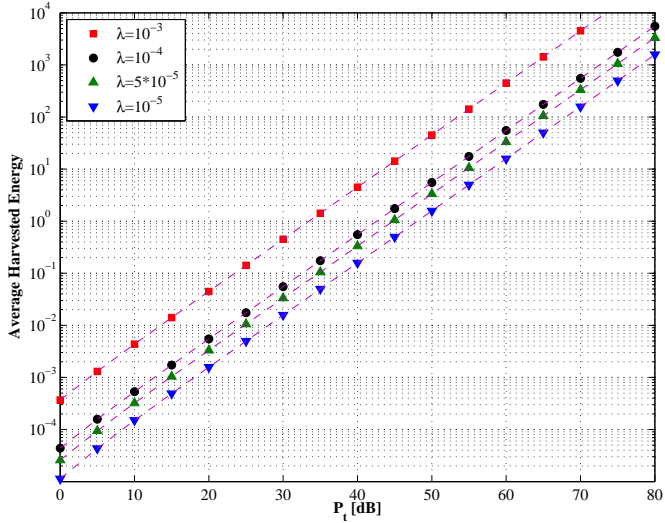}
   \label{fig3b}
 }
\label{figure3}
\vspace{-0.3cm}
\caption{Performance of the cooperative protocol versus $P_t$ for different network densities $\lambda$; $P_r=P_t$, $\lambda_r=10^{-2}$, $\sigma^2=\sigma_C^2=1$, $\Omega=-30$ dB, $r_0=4$ m, $\eta=8$ m, $\theta_0=\pi/3$,   $d_0=20$ m, $\nu_d=\nu_r=0.3$ and $\alpha=4$.
Analytical results are shown with dashed lines.}
\end{figure}

\subsection{Cooperative protocol}

Although our analysis is general and concerns any $\lambda$, $\lambda_r$, in the simulations results,  we assume that $\lambda_r$ is much higher than the network density $\lambda$ in order to demonstrate the potential gains from relaying at the outage/harvesting performance. More specifically, with a small $\lambda_r$ (e.g., $\lambda_r=10^{-4}$ or $10^{-5}$)  the probability of decoding at the relays (according to \eqref{prob_fst}) becomes almost zero and therefore we cannot show the impact of cooperation. On the other hand, as $\lambda$ increases, the multi-user interference significantly increases and the achieved outage probability of the system has not practical interest. Therefore, in order to reveal the potential benefits of cooperation, we assume a small $\lambda$ which ensures a low (non-cooperative) probability outage floor as well as a higher $\lambda_r$ which provides a non-empty relay set and therefore cooperative diversity. A further optimization of the network densities is an interesting  problem that could be considered for future work \cite{LEE}; here, we assume that network densities are fixed and can not be taken into account in the design. This network density setup with $\lambda<\lambda_r$ could refer to a bursty network with a small transmission probability, where relays are part of the same network and correspond to the inactive nodes \cite{MOH}.

Fig.'s \ref{fig3a}, \ref{fig3b} show the performance of the cooperative protocol in terms of outage probability and average harvested energy, for a simulation setup with $P_r=P_t$, $\nu_d=\nu_r=0.3$, $\lambda_r=10^{-2}$,  $\eta=8$ m and $\theta_0=\pi/3$ with $\theta_0\leq 0.4359\pi$ (Remark 1); the other parameters are defined as before. Specifically, Fig. \ref{fig3a} plots the outage probability versus the transmitted power $P_t$. As it can be seen the main observations are similar to the non-cooperative protocol and thus the outage probability converges to a constant outage floor for high $P_t$, since there is not any coordination/scheduling in both phases of the protocol. As for the network density $\lambda$, we can see that it significantly affects the outage performance of the system;  the associated multi-user interference degrades the decoder's performance at both the receivers  and  the relays in the first phase of the protocol. In the same figure, we plot the theoretical expressions given by \eqref{out_expr}; we can see that the theoretical results provide a near-perfect match to the simulations results and validate our analysis for the cooperative case.

A direct comparison between Fig.'s \ref{fig1a}, \ref{fig3a} for high $P_t$ (i.e., $P_t\rightarrow \infty$)\footnote{For the case $P_t\rightarrow \infty$, the comparison between non-cooperative and cooperative protocol is fair and the their performance gap is due to the cooperative diversity associated with the cooperative scheme.}, shows that the cooperative protocol improves the outage probability of the system and achieves a lower outage probability floor e.g., for $\lambda=10^{-5}$, the outage probability converges to an outage probability equal to $6\times 10^{-4}$ and $3\times 10^{-4}$, for the non-cooperative protocol and the cooperative protocol, respectively, (a more significant gain can be observed for another simulation setup as it is reported in the following discussion). In case of a successful decoding at the relay nodes (the relay set is not empty), the cooperative protocol provides a retransmission of the source's message from a shorter distance than the direct link as well as via an independent fading channel and therefore improves the achieved outage probability due to the cooperative diversity. In Fig. \ref{fig3b}, we plot the RF average harvested energy versus the transmitted power $P_t$. As it can be seen the relaying operation further improves the average harvested energy since the receivers can harvest energy from the relaying links. For $\lambda=10^{-3}$ and $P_t=60$ dB, the average harvested energy increases from $275$ Watt to $450$ Watt due to cooperation. It is worth noting that when a receiver has not relay assistance, it uses all the received power during the second phase of the cooperative protocol for energy harvesting.

In Fig. \ref{fig4}, we show the impact of the central angle $\theta_0$ and the relay density $\lambda_r$ on the outage performance of the cooperative protocol; we assume $\lambda_{r}=\{10^{-1},10^{-2},10^{-3}\}$, $\theta_0=\{\pi/3, \pi/2, \pi \}$, while the other simulation parameters follows the previous example. It can be seen that the combination $(\lambda_r,\theta_0)=(10^{-1},\pi/2)$ achieves the best outage performance for high $P_t$ (it converges to the lowest outage floor). This result reveals a very interesting relation between these two parameters as well as a multidimensional trade-off. Specifically, a high $\lambda_r$ ensures a non-empty relay set during the first phase of the cooperative protocol and provides cooperative diversity benefits. In this case, a smaller angle i.e., $\theta_0=\pi/2$ (Remark 1) can further protect the system from large-distance relays and therefore achieves a better outage probability performance than $\theta_0=\pi$. On the other hand, as $\lambda_r$ decreases, the probability to have  a non-empty relay set increases and a larger angle is required in order to still have a potential relay at the area of the transmitter; the condition in Remark 1 becomes less important, since successful relay decoding is the priority for the system.  The combination  $(\lambda_r,\theta_0)=(10^{-1},\pi/2)$ seems to provide the best balance between successful relay decoding and protection from  large-distance relays. In comparison to the non-cooperative protocol, the considered setting reveals a significant gain of the cooperative protocol against the non-cooperative scheme e.g., for $\lambda_r=0.1$ and $\theta_0=\pi/2$, the outage probability converges to $10^{-7}$ in comparison to  $6\times 10^{-4}$ reported in Fig. \ref{fig1a}. Finally, Fig. \ref{fig5} depicts the average harvested energy versus $P_t$.  We observe that a small angle is beneficial for the energy harvesting operation since an empty relay set allows the receiver to use all the received signal for energy harvesting in the second phase of the protocol. In addition, as $\lambda_r$ increases, the probability of relaying increases which is beneficial for the energy harvesting process. 

\begin{figure}[t]
\centering
\includegraphics[width=\linewidth]{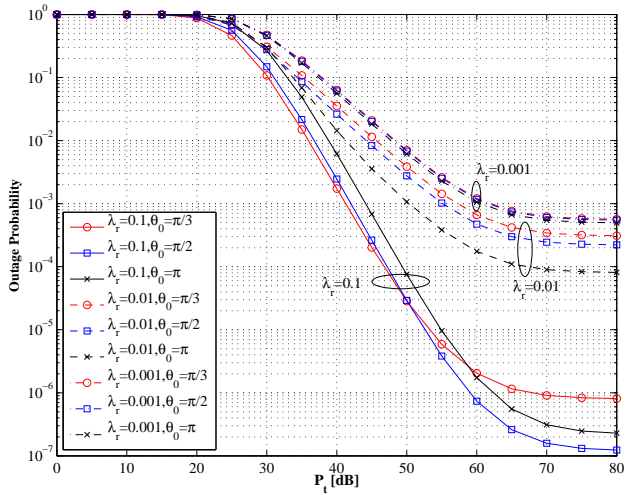}
\vspace{-0.5cm}
\caption{Outage probability versus $P_t$ for different $\lambda_r$ and $\theta_0$; $P_r=P_t$, $\lambda=10^{-5}$, $\lambda_{r}=\{10^{-1},10^{-2},10^{-3}\}$, $\sigma^2=\sigma_C^2=1$,  $\theta_0=\{\pi/3, \pi/2, \pi \}$, $\Omega=-30$ dB, $r_0=4$ m, $\eta=8$ m,   $d_0=20$ m, $\nu_d=\nu_r=0.3$ and $\alpha=4$.}
\label{fig4}
\end{figure}

\begin{figure}[t]
\centering
\includegraphics[width=\linewidth]{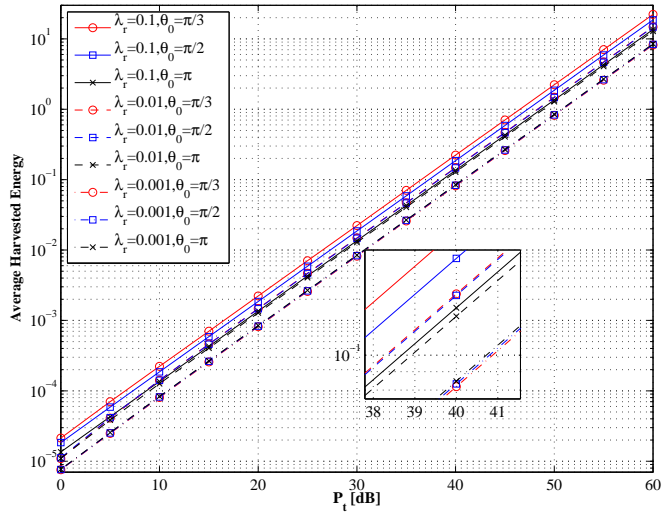}
\vspace{-0.5cm}
\caption{Average harvested energy versus $P_t$ for different $\lambda_r$ and $\theta_0$;  $P_r=P_t$, $\lambda=10^{-5}$, $\lambda_{r}=\{10^{-1},10^{-2},10^{-3}\}$, $\sigma^2=\sigma_C^2=1$,  $\theta_0=\{\pi/3, \pi/2, \pi \}$, $\Omega=-30$ dB, $r_0=4$ m, $\eta=8$ m,   $d_0=20$ m, $\nu_d=\nu_r=0.3$ and $\alpha=4$.}
\label{fig5}
\end{figure}

\section{Conclusion}\label{conc}

This paper has dealt with the PS harvesting technique in large-scale networks with multiple transmitter-receiver pairs, where receivers are characterized by both QoS and RF energy harvesting requirements. A non-cooperative scheme  where all transmitters simultaneously communicate with their associated receivers without any coordination, is analyzed in terms of outage performance and average harvested energy by using stochastic-geometry. We show that network density and power splitting ratio significantly affects the fundamental trade-off between outage performance and energy harvesting. For this case, an optimization problem that minimizes the transmitted power under outage probability and harvesting constraints is formulated and solved in closed form. In addition, a cooperative scheme where sources' transmissions are assisted by a random set of orthogonal relays is analyzed. A random relay selection policy is considered with a sectrorized selection area at the direction of the receivers. Analytical and simulation results reveal the impact of relay density and selection area on the achieved outage-probability/average harvested performance. An extension of this work is to integrate a coordination (scheduling) between the different transmissions and study the trade-off between energy harvesting and potential diversity gains. In addition, more sophisticated cooperative protocols and diversity combining schemes can also be considered in order to further boost the simultaneous information/energy transfer.

%
%
%
%

\appendices

\section{Outage probability for the non-cooperative protocol: Proof of Proposition I}\label{app1}

In order to calculate the outage probability for the non-cooperative protocol, we need to calculate the Laplace transform of the normalized interference term $I_0=\sum_{x\in \Phi_t}d(x)^{-\alpha}$. We have 
\begin{align}
\mathcal{L}_{I_0}(s)&=\mathbb{E}\big(\exp(-s I_0) \big)\nonumber \\
&=\mathbb{E} \left(\exp \left(-s\sum_{x\in \Phi_t}d(x)^{-\alpha} \right) \right) \nonumber \\
&=\mathbb{E} \left( \prod_{x\in \Phi_t}\exp \big(-s d(x)^{-\alpha} \big) \right) \nonumber \\
&=\mathbb{E} \left( \prod_{\substack{x\in \Phi_t, \\ \|x-r(x_0) \|>r_0}}\exp(-s d(x)^{-\alpha}) \right) \nonumber \\
&\;\;\;\;\times \mathbb{E} \left( \prod_{\substack{x\in \Phi_t, \\ \|x-r(x_0) \| \leq r_0}}\exp(-s r_0^{-\alpha}) \right) \nonumber \\
&= \mathcal{L}_{I_0}'(s)  \big[\exp \left(-sr_0^{-\alpha} \right)\big]^{\mathbb{E}\big(N(b(0,r_0)) \big)} \nonumber \\
&= \mathcal{L}_{I_0}'(s) \exp \left(-s  \pi \lambda r_0^{2-\alpha} \right),\label{teliko}
\end{align}

\noindent where $\mathbb{E}\big(N(b(0,r_0)) \big)=\lambda \pi r_0^2$ denotes the average number of points $x_k \in \Phi_t$ falling in a disk of radius $r_0$ \cite[2.4.2]{HAN2}.  For the computation of $\mathcal{L}_{I_0}'(s)$, we have

\begin{subequations}
\begin{align}
\mathcal{L}_{I_0}'(s)&=\mathbb{E} \left( \prod_{\substack{x\in \Phi_t \\ \|x-r(x_0) \| > r_0}}\exp\left(-s d(x)^{-\alpha} \right) \right)  \label{e1} \nonumber \\
&=\exp \bigg(-\lambda \int_{\mathcal{R}} \bigg(1-\exp \left(-s r^{-\alpha} \right) \bigg)dr \bigg)  \\
&=\exp \left(-\lambda \int_{-\pi}^{\pi}\!\int_{r_0^\alpha}^{\infty}\frac{1}{\alpha} \left(1-\exp\left(-\frac{s}{y} \right) \!\right)\! y^{\delta-1}dy d\theta  \right) \label{e2} \\
&=\exp\left(- 2\lambda\pi  \int_{0}^{r_0^{-\alpha}} \frac{1}{\alpha}\bigg(1-\exp(-su)  \bigg)u^{-\delta-1} du \right) \label{e3} \\
&=\exp \Bigg(-\lambda \pi \Bigg[\bigg(\exp \left(-s r_0^{-\alpha} \right)-1\bigg)r_0^2 \nonumber \\
&\;\;\;\;\;+s^{\delta}\gamma \left(1-\delta, s r_0^{-\alpha} \right) \Bigg] \Bigg), \label{e4}
\end{align}
\end{subequations}

\noindent where $\mathcal{R}\triangleq\{ r_0 \leq r,\; \theta \in [-\pi,\pi] \}$ denotes the integration area, \eqref{e1} follows from the probability generating functional of a PPP \cite[Sec. 4.6]{HAN2},  \eqref{e2} by using the transformation $y\leftarrow r^{\alpha}$,  \eqref{e3} by using the transformation $u\leftarrow y^{-1}$, and \eqref{e4} from integration by parts;  $\gamma(n,\beta)\triangleq\int_{0}^{\beta}y^{n-1}\exp(-y)dy$ is the lower incomplete gamma function \cite{GRA}. 

The outage probability for the typical transmitter-receiver link $x_0\rightarrow r(x_0)$ can be written as 

\begin{subequations}
\begin{align}
P_{\text{out}}&=1-\mathbb{P}\left(\frac{\nu_d P_t h_0 d_0^{-\alpha}}{\nu_d(\sigma^2+P_t I_0)+\sigma_C^2} \geq \Omega  \right) \nonumber \\
&=1-\mathbb{P}\left( h_0 \geq \frac{\Omega d_0^{\alpha} \sigma^2}{P_t}+\frac{\Omega d_0^{\alpha} \sigma_C^2}{\nu_d P_t}+\Omega d_0^{\alpha} I_0 \right) \nonumber \\
&=1-\mathbb{E} \exp\left( -\frac{\Omega d_0^{\alpha} \sigma^2}{P_t}-\frac{\Omega d_0^{\alpha} \sigma_C^2}{\nu_d P_t}-\Omega d_0^{\alpha} I_0  \right) \label{p1} \\
&=1- \exp \left( -\frac{\Omega d_0^{\alpha} \sigma^2}{P_t}-\frac{\Omega d_0^{\alpha} \sigma_C^2}{\nu_d P_t}     \right)\underbrace{\mathbb{E} \exp(-\Omega d_0^\alpha I_0)}_{\mathcal{L}_{I_0}(\Omega d_0^\alpha)},  \label{p2}
\end{align}
\end{subequations}

\noindent where \eqref{p1} follows from the cumulative distribution function of an exponential random variable with unit variance $F_X(x)=1-\exp(-x)$ and the Laplace transform in \eqref{p2} is given by \eqref{teliko}. It is worth noting that although the above analytical method is similar to several stochastic geometry works e.g., \cite{WEB1,HAN2}, our analysis/result concerns a different problem and is based on different system assumptions.

\section{Mean of the interference term $I_0$}\label{app2}

Let $\Phi_t$ be a PPP with density $\lambda$ and let $I_0=\sum_{x\in \Phi_t}d(x)^{-\alpha}$;  by using Campbell's theorem for the expectation of a sum over a point process \cite[4.2]{HAN2}, we have: 

\begin{align}
\mathbb{E}(I_0)&=\mathbb{E}\left(\sum_{\substack{x\in \Phi_t, \\ \|x-x_0 \|>r_0}}d(x)^{-\alpha}  \right)+\mathbb{E}\left(\sum_{\substack{x\in \Phi_t \\ \|x-x_0 \|\leq r_0}}r_0^{-\alpha}  \right) \nonumber \\
& =\lambda \int_{-\pi}^{\pi} \int_{r_0}^{\infty}r^{-\alpha} r dr d\theta+ \mathbb{E} \bigg(N(b(0,r_0)) \bigg)r_0^{-\alpha} \nonumber \\
&=\frac{2\pi\lambda r_0^{2-\alpha}}{\alpha-2}+\lambda \pi  r_0^{2-\alpha} \nonumber \\
&=\pi \lambda r_0^{2-\alpha}\frac{\alpha}{\alpha-2},
\end{align}

\noindent where $\mathbb{E} \bigg(N(b(0,r_0)) \bigg)=\lambda \pi r_0^2$.

\section{Selection sector $\mathcal{B}_k$- central angle}\label{angle_sec}

We define as  $\theta\triangleq \angle\; y \widehat{x_k}r(x_k)$ the angle which is formed by the relay node $y$, the transmitter $x_k$ and the receiver $r(x_k)$, $r\triangleq d(x_k,y)$ and $c\triangleq d(y,r(x_k))$, as depicted in Fig. \ref{model2}. By using the cosine rule, the requirement that the relay-receiver distance should be shorter than $d_0$ gives:
\begin{align}
&r^2+d_0^2-2r d_0 \cos\theta\leq d_0^2 \nonumber \\
&\Rightarrow \theta \in \left[-\cos^{-1}\left(\frac{r}{2d_0} \right),\;\cos^{-1}\left(\frac{r}{2d_0} \right)  \right].
\end{align}
In the case where the selection area is a sector with a constant central angle, by applying the above condition to the border of the sector (i.e., for a distance $\eta$), we have
\begin{align}
\theta \in \left[-\cos^{-1}\left(\frac{\eta}{2d_0} \right),\;\cos^{-1}\left(\frac{\eta}{2d_0} \right)  \right].
\end{align}
It is worth noting that the above condition gives the maximum range of the angle; any angle defined in this range, it also supports the distance requirement.

\section{Outage probability for the first hop (empty relay set)- $\Pi_c(P_t)$}\label{app_c1}

Let $r$ be the distance between transmitter and relay.  The relay nodes that are able  to successfully decode the source's signal form the  point process $\Phi_{r}'$, which is generated by the homogeneous PPP process $\Phi_{r}$ by applying a thinning procedure \cite[2.7.3]{HAN2}; therefore $\Phi_{r}'$ is a PPP with intensity 
\begin{align}
\lambda_{r'}(x)&= \lambda_r \mathbb{E}\bigg(\mathbf{1} \big(x_0\rightarrow y_k | \Phi_t \big) \bigg)\\
&=\left\{ \begin{array}{l} \lambda_r \exp\left(-\frac{\sigma^2 \Omega r^{\alpha}}{P_t} \right)\Xi(\lambda,r,r_0)\;\;\;\text{If}\; r> r_0 \\ \lambda_r \exp\left(-\frac{\sigma^2 \Omega r_0^{\alpha}}{P_t} \right)\Xi(\lambda,r_0,r_0)\;\text{If}\; r\leq r_0 ,  \end{array} \right.
\end{align}

\noindent where for the above expression we have used the expression in Proposition \ref{prop1} for a direct distance equal to $r$ and $\sigma_C^2=0$. If we focus on the typical transmitter, the mean of $\Phi_{r}'$ inside the area $\mathcal{B}_0$ is equal to 
\begin{align}
\mu_{r'}(\mathcal{B}_0)&=\int_{\mathcal{B}_0}\lambda_{r'}(x)dx \nonumber \\
&=\int_{\mathcal{B}_0} \lambda_r\! \exp\left(-\frac{\sigma^2 \Omega r^{\alpha}}{P_t} \right)\Xi(\lambda,r,r_0)dx \nonumber \\
&\;\;\;\;\;+\mathbb{E}\big[N(\mathcal{B}_0) \big]\lambda_r \exp\left(-\frac{\sigma^2 \Omega r_0^{\alpha}}{P_t} \right)\Xi(\lambda,r_0,r_0) \nonumber \\
&=\int_{-\theta_0}^{+\theta_0}\!\!\int_{r_0}^{\eta}  \lambda_r \exp\left(-\frac{\sigma^2 \Omega r^{\alpha}}{P_t} \right)\Xi(\lambda,r,r_0)r dr d\theta \nonumber \\
&\;\;\;\;\;+\lambda_r \theta_0 r_0^2 \exp\left(-\frac{\sigma^2 \Omega r_0^{\alpha}}{P_t} \right)\Xi(\lambda,r_0,r_0)
\end{align}

By using fundamental properties of a PPP process \cite[2.4.3]{HAN2},  the probability to have an empty relaying set is equal to
\begin{align}
\Pi_c(P_t)&=\mathbb{P}\{N(\mathcal{B}_0)=0 \} =\exp\big(-\mu_{r'}(\mathcal{B}_0)  \big).
\end{align} 

\section{Outage probability for the relaying hop- $\Pi_r(\nu_r,P_r)$}\label{app_c2}

By using the cosine rule, the distance relay-receiver can be expressed as $c \triangleq \sqrt{r^2+d_0^2-2rd_0 \cos(\theta)}$, where $r$ denotes the distance transmitter-relay (see also Fig. \ref{model2}). In the case of a relaying transmission, the interference at each receiver is generated by all selected relays which form a  PPP $\Phi_{r^*}$ with density $\lambda (1-\Pi_c)$ (i.e., one relay is selected for each transmitter with probability $(1-\Pi_c)$). For the outage probability of the relaying link, we can apply the derived expressions for the direct link as follows
\begin{align}
\Pi_r(\nu_r,P_r)&=1-\frac{1}{|\mathcal{B}_k|}\int_{\mathcal{B}_k}\exp\left(-\frac{\sigma^2 \Omega c^{\alpha}}{P_r} \right) \exp\left(-\frac{\sigma_C^2 \Omega c^{\alpha}}{\nu_r P_r} \right) \nonumber \\
&\;\;\;\;\;\times \Xi(\lambda (1-\Pi_c), c, r_0)dx \nonumber \\
&=1-\frac{1}{\theta_0(\eta^2-r_0^2)}\int_{-\theta_0}^{\theta_0}\int_{r_0}^{\eta}\exp\left(-\frac{\sigma^2 \Omega c^{\alpha}}{P_r} \right)  \nonumber \\
&\;\;\;\;\;\times \exp\left(-\frac{\sigma_C^2 \Omega c^{\alpha}}{\nu_r P_r} \right) \Xi(\lambda (1-\Pi_c), c, r_0)r dr d\theta, 
\end{align}

\noindent where $|\mathcal{B}_k|=\theta_0\eta^2-\theta_0 r_0^2$ gives the area of $\mathcal{B}_k$. 
We note that the above expression takes into account that the smallest distance between a communication pair is $r_0$; the points of $\mathcal{B}_k$ with $r<r_0$ are considered to have a distance $r_0$ according to the considered radio propagation model in \eqref{modelo_apostasis}.

\section{Average relaying attenuation}\label{app_aver}

Let $c=d(y,r(x_k))$ the distance relay-receiver and $r=d(x_k,y)$ the distance transmitter-relay; by using the cosine rule (see also Fig. \ref{model2}), we have $c^2=r^2+d_0^2-2r d_0 \cos(\theta)$. 
For a selection area $\mathcal{B}_k=\{r \in [0\;\eta],\;\theta\in [-\theta_0,\;\theta_0] \}$,  the average relay-receiver attenuation can be expressed as
\begin{align}
&\mathbb{E}(c^{-\alpha})=\mathbb{E}\big( r^2+d_0^2-2r d_0 \cos(\theta)\big)^{-\delta} \nonumber \\
&=\frac{1}{|\mathcal{B}_k|} \int_{\mathcal{B}_k}\big( r^2+d_0^2-2r d_0 \cos(\theta)\big)^{-\delta}dx \nonumber \\
&=\frac{1}{\theta_0 (\eta^2-r_0^2)}\int_{-\theta_0}^{\theta_0}\int_{r_0}^{\eta}\big( r^2+d_0^2-2r d_0 \cos(\theta)\big)^{-\delta}  r dr d\theta,
\end{align} 

\noindent where the above expression takes into account the radio propagation model in \eqref{modelo_apostasis}. In order to have a simple expression for the average relaying attenuation, we apply Jensen's inequality:
\begin{subequations}
\begin{align}
\mathbb{E}(c^{-\alpha})&=\mathbb{E}\big(r^2+d_0^2-2rd_0 \cos(\theta) \big)^{-\delta} \nonumber \\
&\geq \big(\mathbb{E}(r^2+d_0^2-2rd_0 \cos(\theta)) \big)^{-\delta}  \label{a1}\\
&\geq \bigg(  \big(\mathbb{E}(r)\big)^2+d_0^2-2d_0\mathbb{E}(r)\mathbb{E}\big(\cos(\theta) \big)\bigg)^{-\delta}  \label{a2}\\
&=\bigg(\left(\frac{r_0+\eta}{2}\right)^2+d_0^2-2d_0 \frac{r_0+\eta}{2}\cdot \frac{\sin(\theta_0)}{\theta_0}  \bigg)^{-\delta},  \label{a3}
\end{align}
\end{subequations}
where \eqref{a1} holds due to the convexity of the functions $f(x,\theta)=(x^2+d_0^2-2xd_0 \cos(\theta))^{-\delta}$,  \eqref{a2} holds due to the convexity of $f(x)=x^2$ and $\mathbb{E}[\cos(\theta)]=\sin(\theta_0)/\theta_0$ in \eqref{a3}.

\begin{IEEEbiography}[{\includegraphics[width=1in,height=1.25in,clip,keepaspectratio]{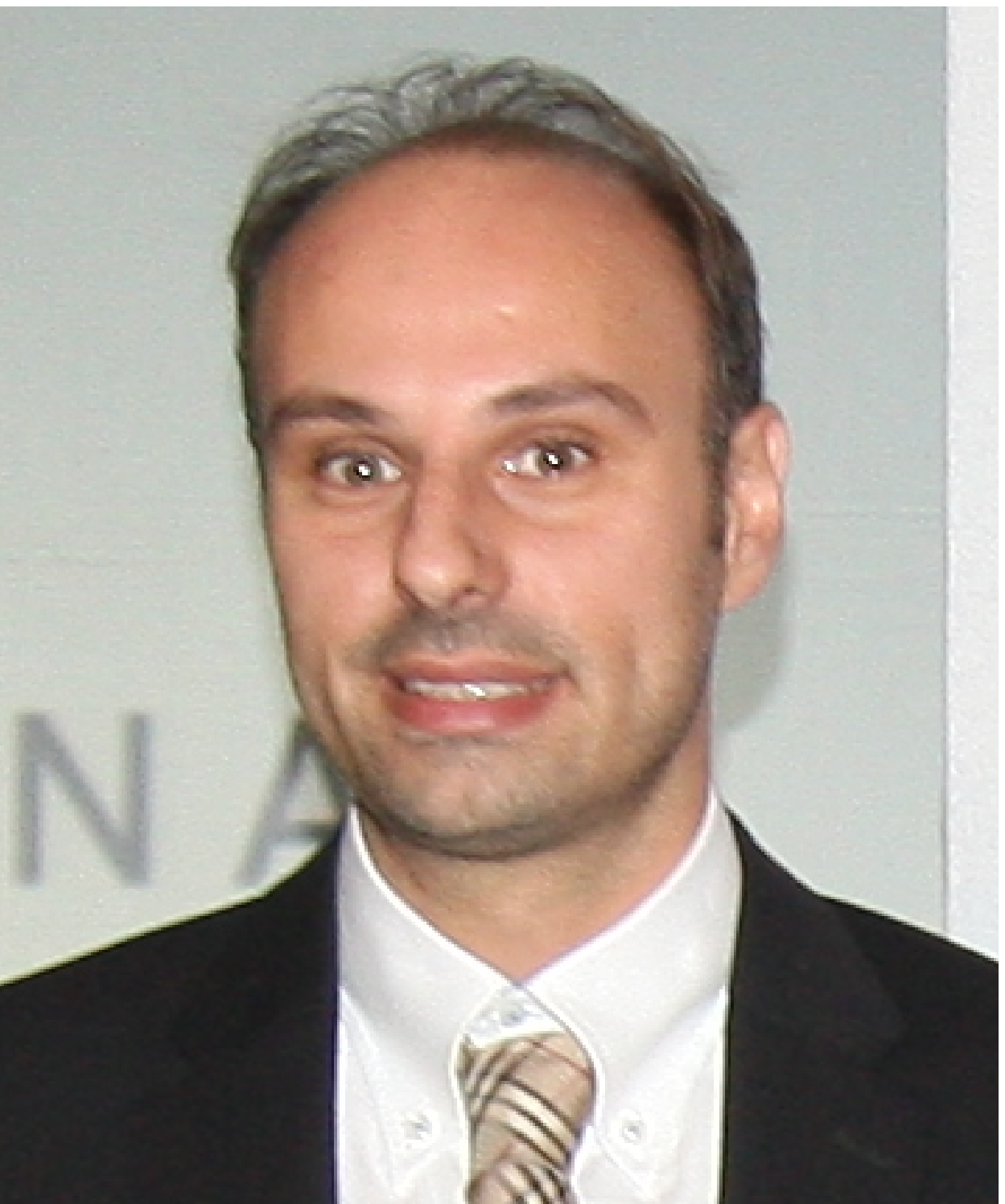}}]{Ioannis Krikidis} (S'03-M'07-SM'12) received the diploma in Computer Engineering from the Computer Engineering and Informatics Department (CEID) of the University of Patras, Greece, in 2000, and the M.Sc and Ph.D degrees from Ecole Nationale Sup\'erieure des T\'el\'ecommunications (ENST), Paris, France, in 2001 and 2005, respectively, all in electrical engineering. From 2006 to 2007 he worked, as a Post-Doctoral researcher, with ENST, Paris, France, and from 2007 to 2010 he was a Research Fellow in the School of Engineering and Electronics at the University of Edinburgh, Edinburgh, UK. He has held also research positions at the Department of Electrical Engineering, University of Notre Dame; the Department of Electrical and Computer Engineering, University of Maryland; the Interdisciplinary Centre for Security, Reliability and Trust, University of Luxembourg; and the Department of Electrical and Electronic Engineering, Niigata University, Japan. He is currently an Assistant Professor at the Department of Electrical and Computer Engineering, University of Cyprus, Nicosia, Cyprus. His current research interests include information theory, wireless communications, cooperative communications, cognitive radio and secrecy communications.

Dr. Krikidis serves as an Associate Editor for the IEEE WIRELESS COMMUNICATIONS LETTERS, IEEE TRANSACTIONS ON VEHICULAR TECHNOLOGY and Elsevier TRANSACTIONS ON EMERGING TELECOMMUNICATIONS TECHNOLOGIES.  He was the Technical Program Co-Chair for the IEEE International Symposium on Signal Processing and Information Technology 2013. He received an IEEE COMMUNICATIONS LETTERS and an IEEE WIRELESS COMMUNICATIONS LETTERS exemplary reviewer certificate in 2012. He was the recipient of the {\it Research Award Young Researcher} from the Research Promotion Foundation, Cyprus, in 2013.
\end{IEEEbiography}

\end{document}